\documentclass[journal,twoside,web]{ieeecolor} 
\usepackage{lcsys}
\usepackage{graphics} % for pdf, bitmapped graphics files
\usepackage{epsfig} % for postscript graphics files
\usepackage{amsmath} % assumes amsmath package installed
\usepackage{amssymb}  % assumes amsmath package installed
 
\usepackage{amsthm}
\usepackage{cite}
\usepackage{color}
\usepackage{enumerate}
\usepackage{graphicx}
\usepackage{mathtools}
\usepackage{cancel}
\usepackage{algorithm}
\usepackage{algpseudocode}
\usepackage{textgreek}
\usepackage{cleveref}
\crefname{equation}{}{}
\crefname{proposition}{Proposition}{Propositions}
\usepackage{url}

\graphicspath{{figures/}}

\newcommand{\real}{\mathbb{R}}

\newcommand{\until}[1]{[#1]}
\newcommand{\map}[3]{#1:#2 \rightarrow #3}

\newcommand{\longthmtitle}[1]{\mbox{}{\textit{(#1):}}}

\newcommand{\setdefb}[2]{\big\{#1 \; | \; #2\big\}}

\newcommand*{\SetSuchThat}[1][]{} % reserve the name
\newcommand*{\MvertSets}{%
    \renewcommand*\SetSuchThat[1][]{%
        \mathclose{}%
        \nonscript\;##1\vert\penalty\relpenalty\nonscript\;%
        \mathopen{}%
    }%
}
\MvertSets % default
\DeclarePairedDelimiterX \Set [2] {\lbrace}{\rbrace}
    {\,#1\SetSuchThat[\delimsize]#2\,}

\newcommand{\Cc}{\mathcal{C}}
\newcommand{\Dc}{\mathcal{D}}

\newcommand{\Ic}{\mathcal{I}}
\newcommand{\Kc}{\mathcal{K}}

\newcommand{\R}{\mathbb{R}}

\newcommand{\Ke}{\Kc^{\rm e}}

\newcommand{\defeq}{\triangleq}

\newtheorem{theorem}{Theorem}

\theoremstyle{definition}
\newtheorem{definition}{Definition}

\newtheorem{remark}{Remark}
\newtheorem{example}{Example}

\renewcommand{\bf}{\mathbf{f}} % NOTE: use \textbf if you really want to write bold non-math text.
\newcommand{\bg}{\mathbf{g}}

\newcommand{\bk}{\mathbf{k}}

\newcommand{\bu}{\mathbf{u}}

\newcommand{\bx}{\mathbf{x}}

\newcommand{\bH}{\mathbf{H}}
\newcommand{\bI}{\mathbf{I}}

\newcommand{\bL}{\mathbf{L}}

\newcommand{\balpha}{\boldsymbol{\alpha}}

\newcommand{\des}{{\operatorname{d}}}

\newcommand{\diag}{{\operatorname{diag}}}

\newcommand{\co}{{\operatorname{co}}}

\renewcommand{\baselinestretch}{1}
\allowdisplaybreaks

\DeclareMathOperator*{\argmin}{argmin}

\title{\LARGE \textbf{Combinatorial Control Barrier Functions: Nested Boolean and $p$-choose-$r$ Compositions of Safety Constraints}}

\author{Pio Ong, Haejoon Lee, Tamas G. Molnar, Dimitra Panagou, Aaron D. Ames %
\thanks{PO and AA are with the Department of Mechanical and Civil Engineering, California Institute of Technology, Pasadena, CA 91125, USA. \texttt{\{pioong, ames\}@caltech.edu}.}
\thanks{HL and DP are with the Department of Robotics, University of Michigan, Ann Arbor, MI 48109, USA. \texttt{\{haejoonl, dpanagou\}@umich.edu}.}
\thanks{TM is with the Department of Mechanical Engineering, Wichita State University, Wichita, KS 67260, USA. \texttt{tamas.molnar@wichita.edu}.}
\thanks{This research was in parts supported by the Technology Innovation Institute and AFOSR Award \#113535-19668.}
}
\begin{document}

\maketitle
\thispagestyle{empty}
\begin{abstract}
    This paper investigates the problem of composing multiple control barrier functions (CBFs)---and matrix control barrier functions (MCBFs)---through logical and combinatorial operations. Standard CBF formulations naturally enable conjunctive (AND) combinations, but disjunctive (OR) and more general logical structures introduce nonsmoothness and possibly a combinatorial blow-up in the number of logical combinations. We introduce the framework of {\em combinatorial CBFs} that addresses $p$-choose-$r$ safety specifications and their nested composition. The proposed framework ensures safety for the exact safe set in a scalable way, using the original number of primitive constraints. We establish theoretical guarantees on safety under these compositions, and we demonstrate their use on a patrolling problem in a multi-agent system.
\end{abstract}

\begin{IEEEkeywords}
   Constrained control, Lyapunov methods, Nonlinear systems, Control barrier functions
    \end{IEEEkeywords}
    
\section{Introduction}

\IEEEPARstart{E}{nsuring} dynamic safety in modern control systems has become essential in applications such as robotics, autonomous vehicles, and aerospace systems. Control barrier functions (CBFs)~\cite{ADA-SC-ME-GN-KS-PT:19} provide one of the most widely used tools for addressing safety.
They were arguably popularized through their ability to be framed as safety filters using the quadratic program (QP) formulation~\cite{ADA-XX-JWG-PT:17}. The resulting optimization-based controllers facilitate the combination of multiple Lyapunov and barrier constraints to handle stability and safety simultaneously. This divide-and-conquer approach has been a key to the practical adoption of CBFs, making it straightforward to handle multiple control criteria.

As safety has taken on greater importance, many works, such as~\cite{MR-MK-SH:16, XX:18,LW-ADA-ME:16}, have extended the QP formulation to integrate multiple CBFs simultaneously. While the optimization framework  facilitates such integration, it also requires a verification on compatibility: even when each CBF admits a safeguarding control, there may not be one in conjunction. This issue has motivated studies on compatibility of multiple CBFs~\cite{XT-DVD:22} and methods for ensuring feasibility of controllers with multiple CBFs~\cite{JB-DP:23,HL-PR-DP:25}. In certain cases, such as parallel safe set boundaries or box constraints, compatibility can be guaranteed~\cite{MHC-EL-ADA:25, KHK-MD-MK:25}. Nevertheless, the majority of these works remains limited to simple conjunctive combinations of CBFs, whereas practical systems often require richer logical structures among safety constraints.

Prior work on the composition of CBFs has largely focused on combining them into a single function and treating the result as a new CBF. For instance, the works~\cite{PG-JC-ME:17, LW-ADA-ME:16} perform AND/OR composition of CBFs through min/max functions and deal directly with the resulting nonsmooth barrier function. This formulation has been applied to marine vehicle applications~\cite{YY-CM-YP:25}. Since nonsmooth functions are difficult to handle in general, other works approximate them with smooth functions such as softmin/softmax~\cite{TGM-ADA:23, TGM:25,MB-DP:23}, enabling applications to quadrotors~\cite{MH-MJ-KA:25} and safe reinforcement learning~\cite{CW-XW-YD-LS-XG:25}. This approach has been later extended to handle hierarchical complex objectives~\cite{RL-ME:25}. Another line of work combines CBFs  through signal temporal logic~\cite{LL-DVD:19,LL-DVD:19b}. More recently, matrix-valued CBFs~\cite{PO-YX-RMB-FJ-ADA:25-tac} have been proposed as a way to capture logical combinations. Despite this progress, the existing literature remains focused on Boolean (AND/OR) combinations of safety constraints. In contrast, here we consider the more general case of combinatorial compositions, safety specifications that commonly arise in fault-tolerant and multi-agent systems (see, e.g., Fig.~\ref{fig:circle}). Our approach builds on the matrix CBF perspective of using multiple inequality constraints rather than a single one. We summarize our contributions next.

\begin{figure}[t]
    \centering
    \vspace{0.28cm}
    \includegraphics[width=\linewidth]{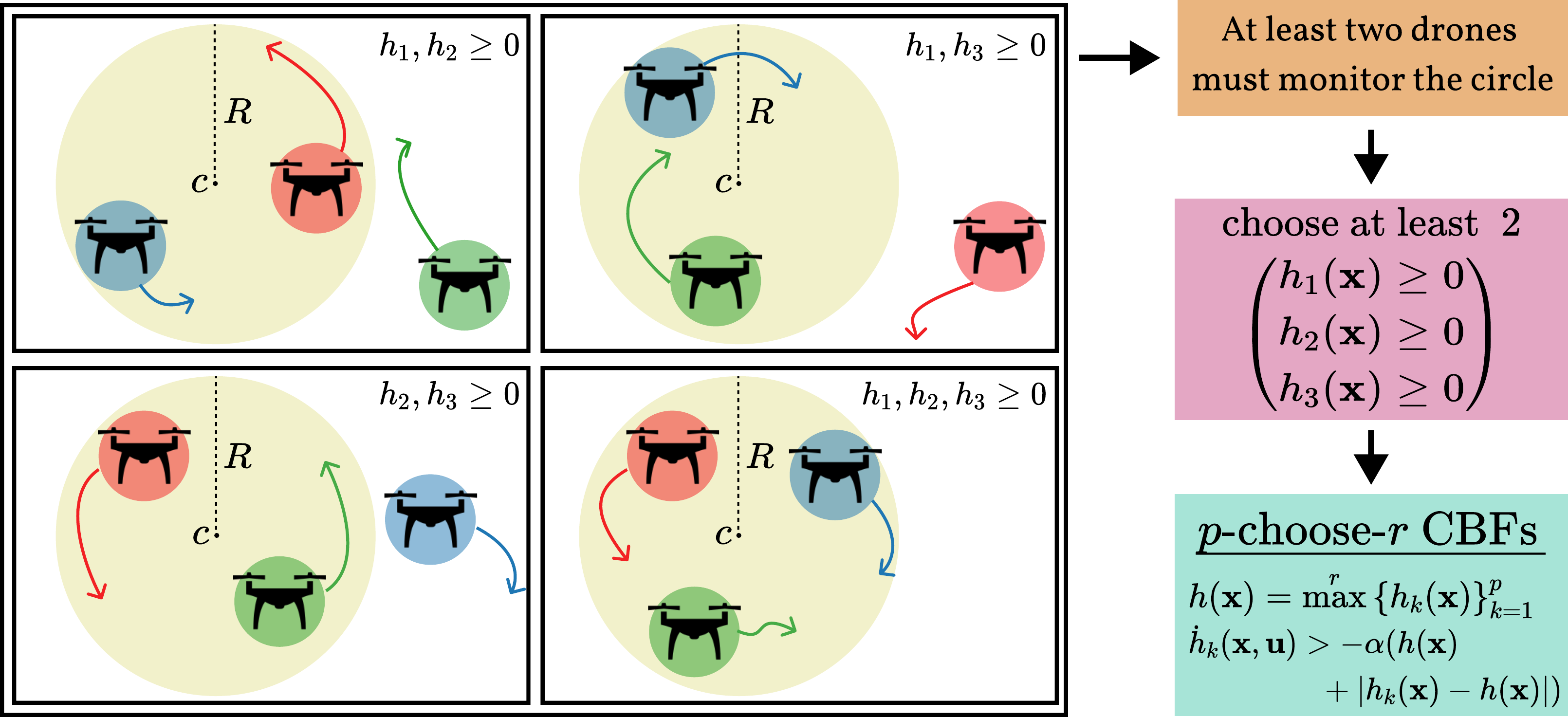}
    \caption{An example of a $p$-choose-$r$ constraint in multi-drone surveillance.
    }
    \label{fig:circle}
\end{figure}

In this paper, we go beyond the Boolean compositions studied in prior works and develop a framework for the combinatorial composition of CBFs. We introduce the notion of $p$-choose-$r$ CBFs that address safety problems where at least $r$ out of $p$ constraints must hold. These CBFs are defined based on
sorting individual primitive constraints and they serve as pivots in the barrier conditions.
Like in the standard AND case, our method enforces safety through multiple inequalities involving the primitive CBFs, which allows us to avoid nonsmoothness issues present in formulations based on nonsmooth min/max operators. In addition, the sorting and pivoting arguments extend naturally to nested logical structures, while ensuring that the number of barrier conditions remains equal to the number of primitive CBFs, despite the combinatorial nature of the safe set composition. Moreover, unlike approaches using smooth relaxation, our method preserves the safe set exactly as specified by the logical combinations. We demonstrate the flexibility and scalability of the framework on a multi-agent patrolling task.

\section{Background}

\subsection{Control barrier functions}
Consider the control-affine system\footnote{For a positive integer $p$, we denote the set of consecutive numbers as ${\until{p}=\{1, 2, \dots, p\}}$. The set of symmetric matrices in $\real^{p \times p}$ is denoted by $\mathbb S^p$, and $\mathbf{I}_p$ is the ${p \times p}$ identity matrix. For a continuously differentiable function ${\map{h}{\real^n}{\real}}$ with a vector field ${\map{\bf}{\real^n}{\real^n}}$, we define the Lie derivative ${L_{\bf}h(\bx) = \frac{\partial h}{\partial \bx}(\bx) \cdot \bf(\bx)}$. For a continuously differentiable matrix-valued function ${\map{\bH}{\real^n}{\mathbb S^p}}$, the Lie derivative ${\bL_{\bf}\bH}$ is defined element-wise: ${L_{\bf}H_{ij}(\bx) = \frac{\partial H_{ij}}{\partial \bx}(\bx) \cdot \bf(\bx)}$ where $H_{ij}$ is the $(i,j)$-th entry of matrix $\bH$. Function ${\map{\alpha}{(-b,a)}{\real}}$, ${a,b>0}$ is of extended class-$\mathcal{K}$ (${\alpha \in \Ke}$) if it is continuous, strictly increasing, and ${\alpha(0)=0}$.}:
\begin{equation}
    \label{sys:ctrl-affine}
    \dot\bx = \bf(\bx)+\bg(\bx)\bu
\end{equation}
with state $\bx\in\real^n$ and control input $\bu\in\real^m$. The system drift $\map{\bf}{\real^n}{\real^n}$ and the control matrix $\map{\bg}{\real^n}{\real^{m\times n}}$ are assumed continuous. Then, if the control signal $t\mapsto \bu(t)$ is continuous, there exists a continuously differentiable solution $t\mapsto \bx(t)$ to the system~\eqref{sys:ctrl-affine}. We are interested in ensuring any solution $\bx(t)$ evolves within a safety constraint~$\Cc\subset\real^n$.

In simpler problems, the safety constraint is given by a single continuously differentiable scalar function ${\map{h}{\real^n}{\real}}$: 
\begin{equation}\label{eq:safeset}
    \Cc = \setdefb{\bx\in\real^n}{h(\bx)\geq 0}.    
\end{equation}
Then, one may address safety using control barrier functions.
\begin{definition}\longthmtitle{Control Barrier Function,~\cite{ADA-XX-JWG-PT:17}}
    A continuously differentiable function $\map{h}{\real^n}{\real}$ is a \textit{control barrier function} (CBF) for~\eqref{sys:ctrl-affine} if there exists $\alpha\in\Ke$ such that, for each $\bx$ in the set $\Cc$ in~\eqref{eq:safeset}, there exists $\bu\in\real^m$ satisfying:
   \begin{equation}\label{eq:CBF}
   \dot h(\bx,\bu) \defeq L_\bf h(\bx) + \sum_{i=1}^m L_{\bg_i}h(\bx) \bu_i
   > -\alpha(h(\bx)).
   \end{equation}
\end{definition}
When $h$ is a CBF, we may find a control signal that keeps the set $\Cc$ \textit{safe}, such that ${\bx(0) \in \Cc \implies \bx(t) \in \Cc}$ for all time, addressing our safety problem\footnote{Generally, the safety problem is addressed by finding a CBF $h$ that defines a safe set $\Cc$ that is a subset of the actual safety constraint.
}.
One way to design a safe control signal is via optimization.
Given $p$ CBFs, denoted by $\{h_k\}_{k=1}^p$, the corresponding constraints can be addressed simultaneously with an optimization-based controller:
\begin{align}\label{eq:CBF-QP}
    \bk(\bx)  =\argmin_{\bu\in\real^m} \quad & \|\bu-\bk_\des(\bx)\|^2                                                                  \\
             \textup{s.t.} \quad & \dot h_k(\bx,\bu) \geq -\alpha(h_k(\bx)),~\forall k\in\until{p}, \nonumber
\end{align}
that changes a desired controller $\bk_\des$ into a safe controller $\bk$.

One advantage of CBFs is their flexibility to account for multiple safety constraints: the CBF-based quadratic programming (CBF-QP) framework above enables a simple integration of multiple CBFs, addressing conjunctive (AND) combination between them.
In this case, we note:  
\begin{equation} \label{eq:cbf_and}
\begin{aligned}
    \!\!\!\bx \!\in\! {\textstyle \bigcap_{k=1}^{p}} & \Cc_k \!\iff\! \bx \!\in\! \Cc_1\ \text{AND } \bx \!\in\! \Cc_2\ \ldots\ \text{AND } \bx \!\in\! \Cc_p, \\
    {\textstyle \bigcap_{k=1}^{p}} & \Cc_k = \setdefb{x \in \real^n}{\min\{h_k(\bx)\}_{k=1}^p \geq 0},
\end{aligned}
\end{equation}
suggesting the formulation renders the intersection of the sets forward invariant, assuming that the CBFs are compatible and the optimization is feasible for all $\bx$ in the set.

This paper investigates logical combinations of constraints beyond AND.
Despite the simplicity in dealing with AND combinations, other types like disjunctive (OR) and more complex logical combinations are difficult to handle.
For example, OR logic encodes the union of safety constraints:
\begin{equation} \label{eq:cbf_or}
\begin{aligned}
    \bx \!\in\! {\textstyle \bigcup_{k=1}^{p}} & \Cc_k \!\iff\! \bx \!\in\! \Cc_1\ \text{OR } \bx \!\in\! \Cc_2\ \ldots\ \text{OR } \bx \!\in\! \Cc_p, \\
    {\textstyle \bigcup_{k=1}^{p}} & \Cc_k = \setdefb{x \in \real^n}{\max\{h_k(\bx)\}_{k=1}^p \geq 0}.
\end{aligned}
\end{equation}
While the AND and OR combinations of safety constraints can be described using the $\min$ and $\max$ functions, these lead to nonsmooth barrier functions~\cite{PG-JC-ME:17} and potentially discontinuous controllers if used directly in optimization.

\subsection{Matrix control barrier functions}

Matrix control barrier functions (MCBFs) can characterize nonsmooth safe sets from logical compositions between CBFs. Consider the safety constraint defined by a continuously differentiable matrix-valued function $\map{\bH}{\real^n}{\mathbb S^p}$:
\begin{align}\label{eq:safeset_AND}
    \Cc&=\setdefb{\bx\in\real^n}{\bH(\bx) \succeq 0}.
\end{align}
In contrast to sets defined by scalar-valued functions, the set above can potentially be nonsmooth. Analogous to CBFs, the previous work~\cite{PO-YX-RMB-FJ-ADA:25-tac} provides the definition for MCBFs.
\begin{definition}\longthmtitle{Matrix CBF,~\cite{PO-YX-RMB-FJ-ADA:25-tac}}
    A continuously differentiable function ${\map{\bH}{\real^n}{\mathbb S^p}}$ is a \textit{matrix control barrier function} (MCBF) for~\eqref{sys:ctrl-affine} if there exists ${\alpha\in \Ke}$ such that, for each $\bx$ in the set $\Cc$ in~\eqref{eq:safeset_AND}, there exists ${\bu\in\real^m}$ satisfying:
   \begin{equation}\label{eq:MCBF}
   \dot\bH(\bx,\bu) \defeq \bL_\bf\bH(\bx) + \sum_{i=1}^m \bL_{\bg_i}\bH(\bx) \bu_i
   \succ -\balpha(\bH(\bx)),
   \end{equation}
   where the matrix function $\map{\balpha}{\mathbb S^p}{\mathbb S^p}$ applies $\alpha$ on the eigenvalues of $\bH$ while keeping the eigenspaces the same.
\end{definition}
We can use MCBFs to enforce safety by posing an optimization problem like~\eqref{eq:CBF-QP}. In this case, the constraint~\eqref{eq:MCBF} (with a positive semidefinite inequality,~$\succeq$) renders the optimization a semidefinite programming (CBF-SDP). 

By enforcing $\bH$ to be positive semidefinite, we effectively ensure that all of its eigenvalues are positive at all times.
We consider the eigenvalues in ascending order:
$$
\lambda_1(\bx) \leq \dots \leq \lambda_p(\bx).
$$
For diagonal matrices, the eigenvalues are the diagonal entries. Thus, when ${\bH(\bx)= \diag(\{h_k(\bx)\}_{k=1}^p)}$ is the diagonalization of $p$ CBFs, the CBF-SDP makes all $\{h_k(\bx)\}_{k=1}^p$ remain nonnegative as in~\eqref{eq:cbf_and}. In this formulation, the constraint derived from the MCBF is equivalent to the inequalities in~\eqref{eq:CBF-QP}, addressing the AND composition between CBFs.

The work~\cite{PO-YX-RMB-FJ-ADA:25-tac} also addresses another type of safe sets:
\begin{align}\label{eq:safeset_OR}
    \Cc=\setdefb{\bx\in\real^n}{\bH(\bx) \not \prec 0}
    =\setdefb{\bx\in\real^n}{\lambda_p(\bx) \geq 0}
\end{align} 
and proposes the following framework.
\begin{definition}\longthmtitle{Indefinite-MCBF, \cite{PO-YX-RMB-FJ-ADA:25-tac}}
A continuously differentiable function ${\map{\bH}{\real^n}{\mathbb S^p}}$ is an \textit{indefinite matrix control barrier function} (Indefinite-MCBF) for~\eqref{sys:ctrl-affine} if there exists ${\alpha\in \Ke}$ and $c_\perp\geq 0$ such that, for each $\bx$ in the set $\Cc$ in~\eqref{eq:safeset_OR}, there exists ${\bu\in\real^m}$ satisfying:
    \begin{equation}\label{eq:Indefinite-MCBF}
   \dot\bH(\bx,\bu) \succ -\alpha\big(\lambda_p(\bx)\big)\bI_{p} -c_\perp\big(\lambda_p(\bx)\bI_{p}\!-\!\bH(\bx)\big).
\end{equation}
\end{definition}
While the motivation of indefinite-MCBF was to address safe sets that are negations of semidefinite constraints, enabling applications like obstacle avoidance, the framework can also address OR combinations between scalar CBFs. When a diagonal matrix ${\bH(\bx)= \diag(\{h_k(\bx)\}_{k=1}^p)}$ is constructed, the Indefinite-MCBF enforces positivity of the maximal CBF $h_{\max}(\bx)= \lambda_p(\bx)$. Furthermore, when the condition in \eqref{eq:Indefinite-MCBF} is diagonal, the corresponding CBF-SDP can be equivalently reformulated into a CBF-QP.

In this paper, we build on the observations on how we can address OR combinations through sorting CBFs, and we generalize the method for addressing \textit{p-choose-r} constraints that, to the best of our knowledge, have not yet been studied.

\section{Combinatorial Control Barrier Functions}
Next, we address combinatorial safety constraints.
Provided $p$ constraints, we enforce that at least $r$ constraints are satisfied at all times (where ${r \leq p}$), regardless of which constraints. We call this a $p$-choose-$r$ safety requirement (with a slight abuse of the more precise term $p$-choose-at-least-$r$). Such compositional structures naturally arise, for example, when defining safe sets through unions and intersections (OR/AND logic, cf. Example~\ref{ex:corner_rect_cross}) or in fault-tolerant and multi-agent settings where only a subset of components must satisfy a safety condition (Example~\ref{ex:circle}). In these cases, simply neglecting the temporarily inactive constraints can cause control discontinuities, motivating our continuous solution below.
Our proposed CBF formulation guarantees the satisfaction of $p$-choose-$r$ conditions while avoiding the combinatorial blow-up in the number of combinations.

Consider $p$ safety constraints in the form  ${h_k(\bx) \geq 0}$, given by functions $\{h_k\}_{k=1}^p$.
To satisfy at least $r$ out of these $p$ constraints, ${h_k(\bx) \geq 0}$ must hold for the $r$-th largest $h_k$ value, which is also the $(p-r+1)$-th smallest $h_k$, denoted~by:
\begin{equation}\label{eq:sorting_CBFs}
    h(\bx)  = \max^r \{h_k(\bx)\}_{k=1}^p = \min^{p-r+1} \{h_k(\bx)\}_{k=1}^p.
\end{equation}
Then, the safe set $\Cc$ is defined as in~\eqref{eq:safeset}.
Note that the special cases $p$-choose-$p$ and $p$-choose-$1$ are equivalent to the AND and OR combination of safety constraints in~\eqref{eq:cbf_and}-\eqref{eq:cbf_or}, that is, $\max^{p}$ and $\max^{1}$ simplify to $\min$ and $\max$, respectively.

\begin{example}\longthmtitle{$p$-choose-$r$ Constraints}
\label{ex:corner_rect_cross}
Figure~\ref{fig:sets} exemplifies safe sets in two dimensions obtained from $p$-choose-$r$ compositions.
The corner in Fig.~\ref{fig:sets}(a) is a $2$-choose-$1$ constraint:
\begin{equation} \label{eq:ex_corner}
\begin{aligned}
    h(\bx) = \max^1\{h_k(\bx)\}_{k=1}^2
    = \max\{h_1(\bx),h_2(\bx)\},
\end{aligned}
\end{equation}
encoding that at least one of the two constraints must be satisfied, hence the safe set is the union of two half spaces.
The rectangle in Fig.~\ref{fig:sets}(b) implies a $4$-choose-$4$ constraint:
\begin{equation} \label{eq:ex_rectangle}
\begin{aligned}
    h(\bx) \!=\! \max^4\{h_k(\bx)\}_{k=1}^4
    \!=\! \min\{h_1(\bx),h_2(\bx),h_3(\bx),h_4(\bx)\},
\end{aligned}
\end{equation}
respecting all four boundaries and creating the intersection of four half spaces.
Meanwhile, a $4$-choose-$3$ composition:
\begin{equation} \label{eq:ex_cross}
\begin{aligned}
    h(\bx) = \max^3\{h_k(\bx)\}_{k=1}^4,
\end{aligned}
\end{equation}
captures the cross-shaped region in Fig.~\ref{fig:sets}(c) where at least three out of four constraints hold.~\hfill $\bullet$
\end{example}

\begin{figure}
\centering
\includegraphics[scale=1]{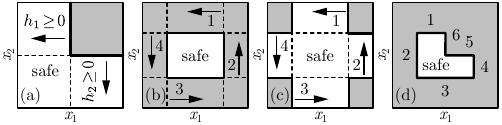}
\vspace{-3mm}
\caption{
Examples of safe sets defined by the logic: (a) $2$-choose-$1$, (b) $4$-choose-$4$, (c) $4$-choose-$3$, and (d) $2$-choose-$2$ of $4$-choose-$4$ and $2$-choose-$1$.
}
\vspace{-5mm}
\label{fig:sets}
\end{figure}

Inspired by the previous result, we establish the concept of combinatorial CBFs to address $p$-choose-$r$ safety constraints.
\begin{definition}\longthmtitle{Combinatorial CBF}
    A function $\map{h}{\real^n}{\real}$ constructed from sorting CBFs $\{h_k\}_{k=1}^p$ as in~\eqref{eq:sorting_CBFs} is a \textit{combinatorial control barrier function} ($p$-choose-$r$ CBF) for~\eqref{sys:ctrl-affine} if there exists $\alpha\in\Ke$ such that, for each $\bx$ in the set $\Cc$ in~\eqref{eq:safeset}, there exists $\bu\in\real^m$ satisfying:
   \begin{equation}
    \dot h_k(\bx,\bu) > -\alpha \big( h(\bx)+\vert h_k(\bx)-h(\bx)\vert \big), \quad \forall k\in\until{p}.
    \label{eq:c-cbf}
    \end{equation}
\end{definition}
The intuition behind the $p$-choose-$r$ CBF $h$ is that we always ensure that the $r$-th largest $h_k$ value, which we refer to as the \textit{pivot}, remains nonnegative. While the absolute value term in~\eqref{eq:c-cbf} allows $h_k$ values smaller than $h$ to become negative, it also guarantees that the $r$-th largest $h_k$, and thus at least $r$ out of the $p$ functions $\{h_k\}_{k=1}^p$, remain nonnegative.

A $p$-choose-$r$ CBF enables the following safety result.
\begin{theorem}\longthmtitle{Combinatorial Safety}
\label{thm:pCr_safety}
    Consider the system~\eqref{sys:ctrl-affine}.
    Let $h$ be constructed from a combination of primitive CBFs $\{h_k\}_{k=1}^p$ as in~\eqref{eq:sorting_CBFs}. If $h$ is a $p$-choose-$r$ CBF for~\eqref{sys:ctrl-affine}, then the set $\Cc$ in~\eqref{eq:safeset} is control invariant (safe).
    
    In addition, any state feedback controller ${\map{\bk}{\real^n}{\real^m}}$,  ${\bu=\bk(\bx)}$, that is continuous and satisfies:
    \begin{equation}
        \label{eq:pCr_BC}
    \dot h_k(\bx,\bk(\bx)) \geq -\alpha \big( h(\bx)+\vert h_k(\bx)-h(\bx)\vert \big), ~\forall k\in\until{p} 
    \end{equation}
    on a neighborhood ${\Dc\supset\Cc}$, renders the set $\Cc$ forward invariant and ensures the bounds:
    \begin{equation}\label{eq:pCr_bound}
        \frac{d}{dt} \big(\max^j\{h_k(\bx(t)\}_{k=1}^p\big) \geq -\alpha\big(\max^j\{h_k(\bx(t)\}_{k=1}^p\big)
    \end{equation}
    for all ${j \geq r}$  at almost every time $t\geq 0$. Consequently, for any initial condition~${\bx_0\in\Cc}$, there exist at least $r$ indices ${k\in \until{p}}$ at almost every time ${t\geq 0}$ such that ${h_k(\bx(t))\geq 0}$.

    In particular, the following CBF-QP:
    \begin{align}\label{eq:pCr-CBF-QP}
    \bk(\bx)  =\argmin_{\bu\in\real^m} \quad & \|\bu-\bk_\des(\bx)\|^2                                                                  \\
             \textup{s.t.} \quad  \dot h_k(\bx,\bu) &\geq -\alpha \big( h(\bx)+\vert h_k(\bx)-h(\bx)\vert \big), ~\forall k\in\until{p} \nonumber
    \end{align}
    is a continuous controller satisfying~\eqref{eq:pCr_BC}.
\end{theorem}
\begin{proof}
    Because each $\max^j\{h_k(\bx)\}_{k=1}^p$ is nonsmooth, the proof relies on nonsmooth barrier function theory. In particular, let $\Kc_j(\bx) = \setdefb{k\in\until p}{h_k(\bx) = \max^j\{h_k(\bx)\}_{k=1}^p}$ be the set of indices of CBFs with the same value as the function of interest.
    Then we have the subgradient~\cite{FHC:83}:
    \begin{align*}
    \partial \big(\max^j\{h_k(\bx)\}_{k=1}^p\big) &= \co\bigcup_{k\in\Kc_j(\bx)} \left\{\frac{\partial h_k}{\partial \bx}(\bx)\right\} \\
    \implies \frac{d}{dt} \big(\max^j\{h_k(\bx(t)\}_{k=1}^p\big)&\in \co\bigcup_{k\in\Kc_j(\bx)}\{\dot h_k(\bx(t),\bu(t))\}.
    \end{align*}
    The implication above follows from the nonsmooth chain rule~\cite[Thm. 2.3.10]{FHC:83} (see also, \cite{PG-JC-ME:17}).
    Then we note the inequalities in~\eqref{eq:pCr_BC} enforce for every $k\in\Kc_j(\bx)$:
    $$
    \dot h_k(\bx(t),\bu(t)) \geq -\alpha(h_k(\bx(t)))
    $$
    when $j\geq r$, from the definition of the sorting operator $\max^j$. Hence, the bounds given in~\eqref{eq:pCr_bound} follow. 
    
    Furthermore, when considering the case $j=r$, we derive:
    $$
    \langle\xi,~\bf(\bx)+\bg(\bx)\bk(\bx)\rangle \geq -\alpha(\max^r\{h_k(\bx)\}_{k=1}^p), 
    $$
    for all $\xi \in \partial\big(\max^r\{h_k(\bx)\}_{k=1}^p\big)$ on the neighborhood $\Dc$. The function $\max^r\{h_k(\bx)\}_{k=1}^p$ is therefore a nonsmooth barrier function~\cite[Prop. 2]{PG-JC-ME:17} for the closed-loop system, verifying the statement of forward invariance.
    
    To prove control invariance, we show that the CBF-QP is a valid safeguarding controller. The definition of the $p$-choose-$r$ CBFs using strict inequality in~\eqref{eq:c-cbf} ensures that the CBF-QP satisfies Slater's condition at every state $\bx$ in a neighborhood $\Dc\supset \Cc$ of the safe set, so it is well-defined and continuous there~\cite{PM-AA-JC:25}. In addition, the CBF-QP enforces~\eqref{eq:pCr_BC} from its constraint, so we may deduce safety.
\end{proof}

Theorem~\ref{thm:pCr_safety} provides the CBF-QP for $p$-choose-$r$ CBFs. An important feature of this construction is that the number of inequalities in the optimization remains $p$ regardless of $r$. This avoids the combinatorial blow-up in the numbers of constraint combinations, yielding a more practical complexity that scales linearly with $p$. This is achieved by sorting the $h_k$ values in~\eqref{eq:sorting_CBFs} to obtain $h$ and then using $h$ as \textit{pivot} in~\eqref{eq:pCr_BC} to enforce constraints only for the $r$ largest $h_k$. The bound~\eqref{eq:pCr_bound} ensures forward invariance with the ordered CBFs rather than the individual primitive CBFs, allowing the active subset of $r$ constraints to change dynamically over time. While it may appear natural to apply the CBF condition directly to the $r$-th largest function (i.e., use ${\dot h \geq -\alpha(h)}$), such an approach involves the derivative of a nonsmooth function, leading to a discontinuous CBF-QP, unlike our proposed framework. Moreover, notice that $p$-choose-$r$ CBFs define exact safe sets, avoiding conservatism introduced by smooth approximations of $h$ using softmin or softmax operators as done in~\cite{TGM-ADA:23, TGM:25, MB-DP:23}. 

Similar to the case of multiple CBFs in~\eqref{eq:CBF-QP}, the feasibility of constraint~\eqref{eq:pCr_BC} is not automatically guaranteed for CBFs $h_k$, but it requires $h$ to be a $p$-choose-$r$ CBF. For instance, in the AND setting ($p$-choose-$p$), the existence of a safe input $\bu$ is only guaranteed individually for each CBF~$h_k$, while joint feasibility needs to be verified, i.e., $h$ must be a $p$-choose-$p$ CBF. This remains an active area of research in the CBF literature. In the $p$-choose-$r$ setting, the situation is subtler: condition~\eqref{eq:c-cbf} must hold even when some ${h_k(\bx)<0}$, i.e., outside the associated set. There, the absolute value term helps relax this requirement. In principle, the absolute value can be replaced by any positive definite function to help address feasibility and compatibility issues, but this choice recovers the standard AND case when ${h(\bx) = \min\{h_k(\bx)\}_{k=1}^p}$ and yields the bounds in~\eqref{eq:pCr_bound}. A full characterization of feasibility is a part of our future work.

\begin{remark}[Task Assignment in Multi-Agent Systems]
A key motivation for $p$-choose-$r$ CBFs arises in multi-agent systems. In many settings, safety constraints apply to each agent independently, but the mission requirements involve only subsets of the agents. For instance, in a surveillance task, the collective objective may be to ensure that a sensitive region is not left unattended, and it is sufficient that at least $r$ agents  remain within the region. Indeed, while mission planning typically handles such task allocation decisions, embedding this logic into combinatorial CBFs enhances safety assurance at the faster control layer, complementing the slower and more uncertainty-prone planning layer. In addition, unlike the single-agent case, feasibility and compatibility are less problematic here since the agents' dynamics and control inputs are decoupled, allowing the constraints to be satisfied in parallel. The following example uses a $p$-choose-$r$ CBF in a multi-agent task assignment scenario.~\hfill $\bullet$

\begin{example}\longthmtitle{Surveillance Task}
\label{ex:circle}
    Consider a circular region shown in Fig.~\ref{fig:circle} and a system of ${p=3}$ robots with dynamics:
    \begin{equation*}
        \dot \bx_k = \bu_k, \ k \in [p].
    \end{equation*}
    We can use a $p$-choose-$r$ CBF to assign at least ${r=2}$ robots out of ${p=3}$ to remain inside the circle at all times. Define the state ${\bx=[\bx_1^\top\ \bx_2^\top\ \ldots\ \bx_p^\top]^\top}$ and the following functions:
    \begin{equation*}
    \begin{aligned}
            h_k(\bx) &= R^2 - \|c-\bx_k\|^2, \ k\in[p], \\
    h(\bx) &= \max^r\{h_k(\bx)\}_{k=1}^p.
    \end{aligned}
    \end{equation*}
    where $R$ and $c$ are the radius and the center of the region.
Each function $h_k(\bx)\geq 0$ indicates that the robot $k$ lies inside the circle, so $h(\bx)$ is a $p$-choose-$r$ CBF which is nonnegative if and only if at least $r$ robots are inside the circle.~\hfill $\bullet$ 
\end{example}
\end{remark}
\section{Nested Logic for Safety Constraints}
We now extend our framework to handle nested combinatorial safety constraints, where multiple combinatorial CBFs are composed hierarchically. This generalization builds on the same sorting and pivoting arguments used for $p$-choose-$r$ CBFs, but applied recursively to follow the nested structure. The resulting formulation encodes the overall safety requirement through a single pivot safety function $h$, which plays the role of the effective barrier function for the nested logic.

For simplicity, we first focus on two-level nested logic. Consider a $p^2$-choose-$r^2$ combination of multiple $p^1_i$-choose-$r^1_i$ combinations among some $i$-th subset of \textit{primitive} CBFs $\{h_k\}_{k=1}^p$. Note carefully  that the superscripts on $p$'s and $r$'s  indicate the nested level of the logical combinations, rather than exponentiation. We define the safe set $\Cc$ as in \eqref{eq:safeset} with:
\begin{equation}
\begin{aligned}
    h_i^1(\bx) & = \max^{r^1_i}\{h_k(\bx)\}_{k\in\Ic_i},
    \quad i \in \until{p^2}, \\
    h(\bx) & =
    \max^{r^2} \{h_k^{1}(\bx)\}_{k=1}^{p^2},
\end{aligned}
\end{equation}
where $\Ic_i$ denote the indices of the $p^1_i$ constraints that are combined through the $p^1_i$-choose-$r^1_i$ logic. Here, we reveal that the nested logic is a logical combination of combinatorial CBFs $\{h_k^1\}_{k=1}^{p^2}$. We provide the following  concrete example.

\begin{example}\longthmtitle{Nested Constraints}
\label{ex:L-shape}
Consider the L-shaped region in Fig.~\ref{fig:sets}(d) that is the intersection of the rectangle in Fig.~\ref{fig:sets}(b) and the corner in Fig.~\ref{fig:sets}(a).
The L-shape is described by a $2$-choose-$2$ (AND) combination of a $4$-choose-$4$ constraint (rectangle) and a $2$-choose-$1$ (corner):
\begin{align}
    h_1^1(\bx) &= \max^4\{h_k(\bx)\}_{k=1}^4,~h_2^1(\bx) = \max^1\{h_k(\bx)\}_{k=5}^6, \nonumber \\
    h(\bx) &= \max^2 \{h_k^{1}(\bx)\}_{k=1}^{2}, \label{eq:two_nested_CBF}
\end{align}
see~\eqref{eq:ex_rectangle} and~\eqref{eq:ex_corner} in~\Cref{ex:corner_rect_cross}.~\hfill $\bullet$
\end{example}

Naively, one would deal with the two-level nested logic with inequalities making $h$ a $p^2$-choose-$r^2$ CBF of $\{h_k^{1}\}_{k=1}^{p^2}$:
$$
    \dot h^1_k(\bx,\bu) > -\alpha \big( h(\bx)+\vert h^1_k(\bx)-h(\bx)\vert \big), \quad \forall k\in\until{p^2}.
$$
However, the main concern here is the nonsmoothness associated with $h_k^1$. We can avoid this issue by applying the barrier condition in~\eqref{eq:pCr_BC} directly on the primitive functions, like for one-level logical combinations. Intuitively, while $h_k^1$ is nonsmooth, its derivative $\dot h_k^1$ remains a convex combination of some primitive $\dot h_k$'s. Therefore, by enforcing~\eqref{eq:pCr_BC} on all primitives, we implicitly enforce the desired inequalities above, rendering $h$ a $p^2$-choose-$r^2$ CBF of $\{h^1_k\}_{k=1}^{p^2}$. In other words, the one-level barrier conditions in~\eqref{eq:pCr_BC} extend directly to the two-level case, with the only change being how $h$ is defined. This observation generalizes by induction to any number of nested logical levels.

Consider the case when $p$ primitive CBFs $\{h_k\}_{k=1}^p$ are combined into a single constraint through $M$ levels of safety specifications with index ${\ell \in \until{M}}$.
At each level, the constraints are combined via $p_i^\ell$-choose-$r_i^\ell$ logic to create $p^\ell$ new safety constraints, described by functions $\{h_i^\ell\}_{i=1}^{p^\ell}$:
\begin{equation}\label{eq:nested_CBF}
\begin{aligned}
    h_i^0(\bx) & = h_i(\bx),
    \quad i \in \until{p}, \\
    h_i^\ell(\bx) & =
    \max^{r_i^\ell} \{h_k^{\ell-1}(\bx)\}_{k \in I_i^\ell},
    \quad i \in \until{p^\ell}, \quad \ell \in \until{M}, \\
    h(\bx) & = h_1^M(\bx),
\end{aligned}
\end{equation}
where $I_i^\ell$ denote the indices of constraints that are combined to obtain constraint $i$ at level $\ell$, with ${I_i^\ell \subseteq \until{p^{\ell-1}}}$, ${|I_i^\ell| = p_i^\ell \leq p^{\ell-1}}$ for all ${i \in \until{p^\ell}}$, ${\ell \in \until{M}}$.
With the function $h$ obtained at the last level, the safe set $\Cc$ is defined in~\eqref{eq:safeset}. We summarize our findings with the following theorem. 

\begin{theorem}
    \longthmtitle{Nested Logic Safety}
    Consider the system~\eqref{sys:ctrl-affine}. Let $h$ be constructed from a multi-level nested logical combination of primitive CBFs $\{h_k\}_{k=1}^p$ as in~\eqref{eq:nested_CBF}. If $h$ is a combinatorial CBF for~\eqref{sys:ctrl-affine}, then the set $\Cc$ in~\eqref{eq:safeset} is safe.

    In addition, any state feedback controller ${\map{\bk}{\real^n}{\real^m}}$,  ${\bu=\bk(\bx)}$, that is continuous and satisfies
    \eqref{eq:pCr_BC} on a neighborhood $\Dc\supset \Cc$ renders the set $\Cc$ forward invariant. The CBF-QP~\eqref{eq:pCr-CBF-QP} is one such controller.~\hfill $\blacksquare$
    \label{thm:nested_logic}
\end{theorem}

We omit the proof for brevity. It follows the same structure as Theorem~\ref{thm:pCr_safety} and proceeds by induction on the nesting levels. The key idea is that the subgradient set can eventually be expressed in terms the gradients of the primitive CBFs. Hence, nested logic introduces no additional complexity beyond redefining the pivot function $h$.

We emphasize that, despite encoding nested logical safety requirements, our formulation requires only the original $p$  primitive constraints. This eliminates the combinatorial blow-up in the number of combinations and significantly reduces the problem size compared to a naive implementation.

\begin{remark}[Logical Composition of MCBFs]
We have described logical compositions starting from primitive scalar CBFs. The same ideas naturally extend to MCBFs, which include scalar CBFs as a special case. In this setting, the role of sorting primitive functions is played by sorting the eigenvalues of the primitive MCBFs, $\{\bH_k\}_{k=1}^p$. In this framework, each MCBF can be interpreted as the first level in a nested CBF: a MCBF corresponds to enforcing a minimum (AND) over its eigenvalues, while an indefinite-MCBF corresponds to enforcing a maximum (OR). A detailed analysis and full technical treatment of this setting are left for future work. ~\hfill $\bullet$
\end{remark}
\section{Simulation}
In this section, we illustrate our results using a multi-agent patrolling problem.
Consider two separated L-shaped regions, ${L_1,L_2 \subset \R^2}$, as shown in Fig.~\ref{fig:L-shape}, that are monitored by ${N=11}$ agents. Each agent follows the dynamics:
\begin{equation} \label{eq:single_integrator}
\dot \bx_j = \bu_j, \ j\in [N]
\end{equation}
and is assigned a desired patrolling controller:
\begin{equation}
\bk_{\text{d},j}(\bx,t) =
\begin{bmatrix}
  \kappa_j\!\bigl(x_{\mathrm{d},j}(t) - x_j\bigr) + \dot{x}_{\mathrm{d},j}(t), \ 0
  \end{bmatrix}^\top,
\end{equation}
where $x_{\mathrm{d},j}(t) = A_j \sin(\omega_j t)$ is the desired patrolling trajectory, with $\kappa_j, A_j, \omega_j > 0$ denoting the proportional gain, motion amplitude, and patrolling frequency, respectively. Here, heterogeneity among the agents is abstracted by varying their patrolling frequency, reflecting their differences in capability.

To ensure robust monitoring against errors in sensor measurements, agent failures, etc., we require at least four agents to remain within each region $L_1$ and $L_2$ at all times. To enforce such requirement using Boolean compositions, we would need to consider a total of ${\sum_{j=4}^{11}{11 \choose j}\sum_{i=1}^2{2 \choose i}=1816\times 3 = 5448}$ combinations, which would not be computationally practical. Therefore, we define the following functions based on the proposed method:
\begin{equation}
\begin{aligned}
h_{L_1}(\bx_j) &= \min\Big\{\min\{h_k(\bx_j)\}_{k=1}^4,~\max \{h_k(\bx_j)\}_{k=5}^6\Big\}, \\
h_{L_2}(\bx_j) &= \min\Big\{\min\{h_k(\bx_j)\}_{k=7}^{10},~\max \{h_k(\bx_j)\}_{k=11}^{12}\Big\},\\
h(\bx) &= \min\Big\{\max^4\{h_{L_1}(\bx_j)\}_{j=1}^{11},~\max^4\{h_{L_2}(\bx_j)\}_{j=1}^{11}\Big\}. \label{eq:ccbf}
\end{aligned}
\end{equation}

\begin{figure}
    \centering
\includegraphics[width=\linewidth]{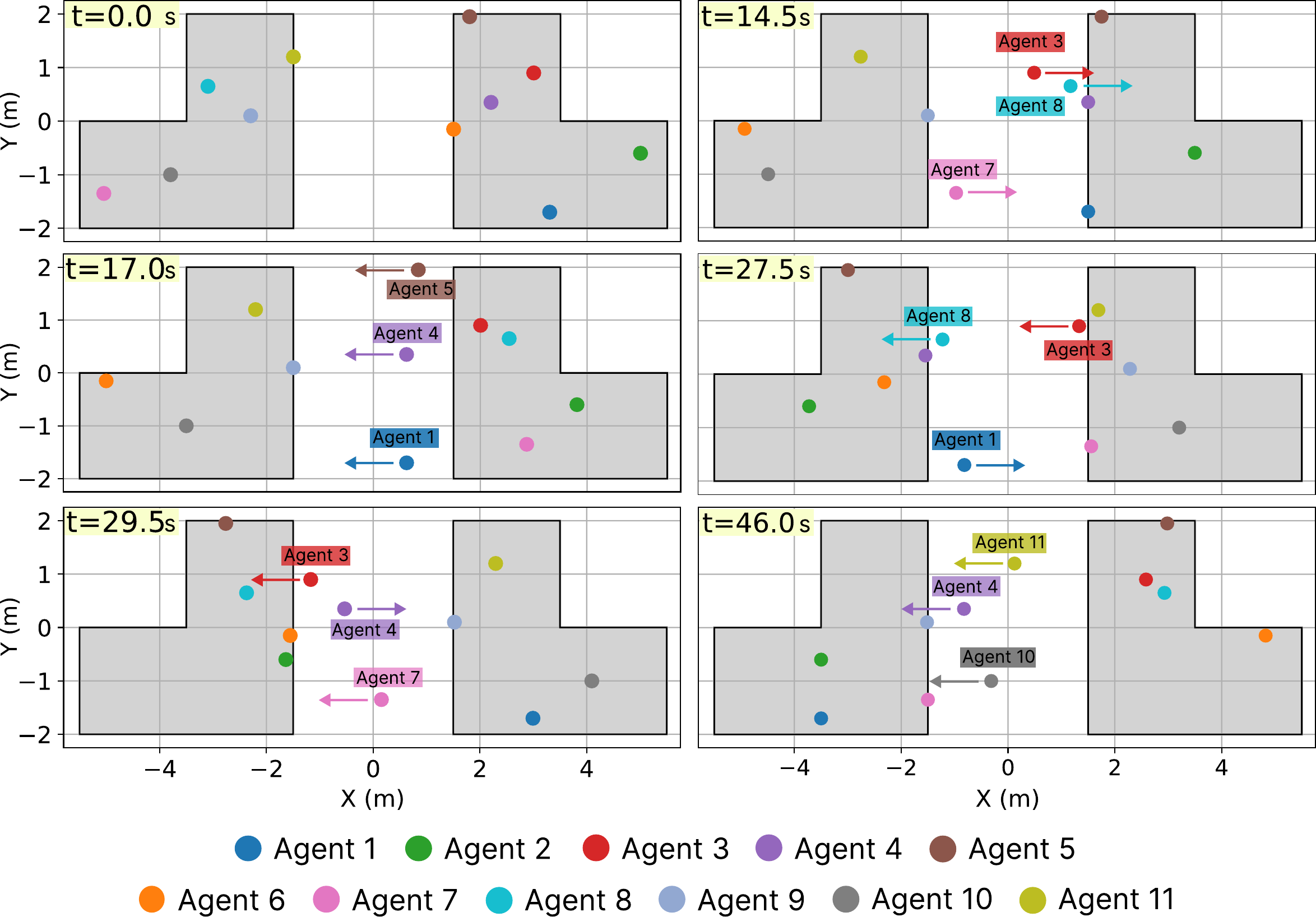}
    \caption{Simulation of a multi-agent patrolling problem where each L-shaped region must be monitored by at least four agents at all times.
    Using the proposed CBF-QP~\eqref{eq:pCr-CBF-QP}, this combinatorial safety constraint is ensured in a computationally tractable manner. The code and animation are available at \protect\url{https://github.com/joonlee16/combinatorial_cbf}.}
    \label{fig:L-shape}
    \vspace*{-5mm}
\end{figure}
Here, $h_{L_1}$ and $h_{L_2}$ characterize the two L-shaped regions using~\Cref{eq:two_nested_CBF} from~\Cref{ex:L-shape} (a 2-choose-2 combination of 4-choose-4 and 2-choose-1). If ${h_{L_1}(\bx_j)\geq 0}$ or ${h_{L_2}(\bx_j)\geq 0}$, then agent $j$'s position resides within the $L_1$ or $L_2$ region, respectively. Using 11 instances of $h_{L_1}$ and $h_{L_2}$ (one per agent), we construct the combinatorial CBF $h$ as given in~\Cref{eq:ccbf} (a 2-choose-2 combination of 11-choose-4). The condition ${h(\bx)\geq 0}$ ensures that at least four agents are in each region simultaneously. By Theorem~\ref{thm:nested_logic}, we can enforce safety through the CBF-QP~\eqref{eq:pCr-CBF-QP} with only ${p=132}$ constraints. 

We simulated the motions of the agents using the dynamics~\eqref{eq:single_integrator}, the CBF-QP~\eqref{eq:pCr-CBF-QP}, and the $p$-choose-$r$ CBF~\eqref{eq:ccbf}. Snapshots for ${t\in[0,50]}$ are shown in Fig.~\ref{fig:L-shape}, with colored arrows indicating agents moving into and out of the L-shaped regions. Because each region must always be covered by at least four agents, at most three agents can be outside both regions at any time. As shown in the figure, this constraint is satisfied throughout the simulation.

\section{conclusion}
We have developed a framework for the logical composition of CBFs and MCBFs that extends beyond simple AND formulations. The proposed combinatorial CBFs not only address $p$-choose-$r$ safety constraints but also their nested logical combinations. Similar to the standard AND case, our formulation enforces safety through multiple inequalities, thereby avoiding the nonsmoothness issues associated with logical operations. The key idea is to correctly identify the safety function via a sorting-based argument and utilize it as a pivot in the barrier conditions. Our future work will explore how to handle the feasibility issue arising when CBF conditions must hold outside their original safe sets and the compatibility issue arising in the intersection of their sets.

\bibliographystyle{ieeetr}
\bibliography{
    bib/alias,
    bib/PO,
    bib/main-Pio
}

\begin{thebibliography}{10}

\bibitem{ADA-SC-ME-GN-KS-PT:19}
A.~D. Ames, S.~Coogan, M.~Egerstedt, G.~Notomista, K.~Sreenath, and P.~Tabuada,
  ``Control barrier functions: Theory and applications,'' in {\em {E}uropean
  {C}ontrol {C}onference}, (Naples, Italy), pp.~3420--3431, June 2019.

\bibitem{ADA-XX-JWG-PT:17}
A.~D. Ames, X.~Xu, J.~W. Grizzle, and P.~Tabuada, ``Control barrier function
  based quadratic programs for safety critical systems,'' {\em IEEE
  Transactions on Automatic Control}, vol.~62, no.~8, pp.~3861--3876, 2017.

\bibitem{MR-MK-SH:16}
M.~Rauscher, M.~Kimmel, and S.~Hirche, ``Constrained robot control using
  control barrier functions,'' in {\em IEEE/RSJ International Conference on
  Intelligent Robots and Systems}, (Daejeon, South Korea), pp.~279--285, Oct.
  2016.

\bibitem{XX:18}
X.~Xu, ``Constrained control of input--output linearizable systems using
  control sharing barrier functions,'' {\em Automatica}, vol.~87, pp.~195--201,
  2018.

\bibitem{LW-ADA-ME:16}
L.~Wang, A.~D. Ames, and M.~Egerstedt, ``Multi-objective compositions for
  collision-free connectivity maintenance in teams of mobile robots,'' in {\em
  {IEEE} Conference on Decision and Control}, (Las Vegas, NV), pp.~2659--2664,
  Dec. 2016.

\bibitem{XT-DVD:22}
X.~Tan and D.~V. Dimarogonas, ``Compatibility checking of multiple control
  barrier functions for input constrained systems,'' in {\em {IEEE} Conference
  on Decision and Control}, (Canc{\'{u}}n, Mexico), pp.~939--944, Dec. 2022.

\bibitem{JB-DP:23}
J.~Breeden and D.~Panagou, ``Compositions of multiple control barrier functions
  under input constraints,'' in {\em {A}merican {C}ontrol {C}onference}, (San
  Diego, CA), pp.~3688--3695, May 2023.

\bibitem{HL-PR-DP:25}
H.~Lee, P.~Rousseas, and D.~Panagou, ``Constraint selection in
  optimization-based controllers,'' {\em arXiv preprint}, no.~2505.05502, 2025.

\bibitem{MHC-EL-ADA:25}
M.~H. Cohen, E.~Lavretsky, and A.~D. Ames, ``Compatibility of multiple control
  barrier functions for constrained nonlinear systems,'' in {\em {IEEE}
  Conference on Decision and Control}, (Rio de Janeiro, Brazil), Dec. 2025.
\newblock To appear.

\bibitem{KHK-MD-MK:25}
K.~H. Kim, M.~Diagne, and M.~Krsti{\'{c}}, ``Constant-sum high-order barrier
  functions for safety between parallel boundaries,'' {\em IEEE Control Systems
  Letters}, vol.~9, pp.~1447--1452, 2025.

\bibitem{PG-JC-ME:17}
P.~Glotfelter, J.~Cort\'es, and M.~Egerstedt, ``Nonsmooth barrier functions
  with applications to multi-robot systems,'' {\em IEEE Control Systems
  Letters}, vol.~1, no.~2, pp.~310--315, 2017.

\bibitem{YY-CM-YP:25}
Y.~Yang, C.~Manzie, and Y.~Pu, ``A control barrier function composition
  approach for multiagent systems in marine applications,'' {\em IEEE/ASME
  Transactions on Mechatronics}, pp.~1--12, 2025.

\bibitem{TGM-ADA:23}
T.~G. Molnar and A.~D. Ames, ``Composing control barrier functions for complex
  safety specifications,'' {\em IEEE Control Systems Letters}, vol.~7,
  pp.~3615--3620, 2023.

\bibitem{TGM:25}
T.~G. Molnar, ``Navigating polytopes with safety: A control barrier function
  approach,'' in {\em {IEEE} Conference on Control Technology and
  Applications}, (San Diego, CA), pp.~179--184, Aug. 2025.

\bibitem{MB-DP:23}
M.~Black and D.~Panagou, ``Adaptation for validation of consolidated control
  barrier functions,'' in {\em {IEEE} Conference on Decision and Control},
  (Marina Bay Sands, Singapore), pp.~751--757, Dec. 2023.

\bibitem{MH-MJ-KA:25}
M.~Harms, M.~Jacquet, and K.~Alexis, ``Safe quadrotor navigation using
  composite control barrier functions,'' in {\em {IEEE} International
  Conference on Robotics and Automation}, (Atlanta, GA), May 2025.

\bibitem{CW-XW-YD-LS-XG:25}
C.~Wang, X.~Wang, Y.~Dong, L.~Song, and X.~Guan, ``Multi-constraint safe
  reinforcement learning via closed-form solution for log-sum-exp approximation
  of control barrier functions,'' in {\em Proceedings of Machine Learning
  Research}, vol.~283, pp.~698--710, June 2025.

\bibitem{RL-ME:25}
R.~Lin and M.~Egerstedt, ``“{H}ierarchy of needs” for robots: Control
  synthesis for compositions of hierarchical, complex objectives,'' in {\em
  {IEEE} International Conference on Robotics and Automation}, (Atlanta, GA),
  pp.~7682--7688, May 2025.

\bibitem{LL-DVD:19}
L.~Lindemann and D.~V. Dimarogonas, ``Control barrier functions for signal
  temporal logic tasks,'' {\em IEEE Control Systems Letters}, vol.~3, no.~1,
  pp.~96--101, 2019.

\bibitem{LL-DVD:19b}
L.~Lindemann and D.~V. Dimarogonas, ``Control barrier functions for multi-agent
  systems under conflicting local signal temporal logic tasks,'' {\em IEEE
  Control Systems Letters}, vol.~3, no.~3, pp.~757--762, 2019.

\bibitem{PO-YX-RMB-FJ-ADA:25-tac}
P.~Ong, Y.~Xu, R.~M. Bena, F.~Jabbari, and A.~D. Ames, ``Matrix control barrier
  functions,'' {\em IEEE Transactions on Automatic Control}, 2025.
\newblock Submitted.

\bibitem{FHC:83}
F.~H. Clarke, {\em Optimization and Nonsmooth Analysis}.
\newblock Canadian Mathematical Society Series of Monographs and Advanced
  Texts, New York: Wiley, 1983.

\bibitem{PM-AA-JC:25}
P.~Mestres, A.~Allibhoy, and J.~Cortés, ``Regularity properties of
  optimization-based controllers,'' {\em European Journal of Control}, vol.~81,
  p.~101098, 2025.

\end{thebibliography}
\end{document}